\documentclass[prodmode,acmtist]{acmsmall} 

\usepackage{times}

\usepackage{latexsym,bm}
\usepackage[fleqn]{amsmath}
\usepackage{float}
\usepackage{amsfonts}
\usepackage{fancyhdr}
\usepackage{amssymb}

\acmVolume{2}
\acmNumber{3}
\acmArticle{21}
\acmYear{2011}
\acmMonth{5}

\begin{document}

\markboth{B. Li et al.}{A Proof of CORN's Universal Consistency: an Online Appendix}

\title{CORN: Correlation-Driven Nonparametric Learning Approach for Portfolio Selection -- an Online Appendix}

\author{BIN LI
\affil{Nanyang Technological University}
DINGJIANG HUANG
\affil{East China University of Science and Technology and Fudan University}
STEVEN C.H. HOI
\affil{Nanyang Technological University}
}

\begin{abstract}
This appendix proves CORN's universal consistency.
One of Bin's PhD thesis examiner (Special thanks to Vladimir Vovk from Royal Holloway, University of London) suggested that CORN is universal and provided sketch proof of Lemma~\ref{lemma:similarity}, which is the key of this proof.
Based on the proof in~\citeN{GLU06}, we thus prove CORN's universal consistency. Note that the notations in this appendix follows~\citeN{GLU06}.
\end{abstract}

\acmformat{Li, B., Hoi, S. C. H., and Gopalkrishnan, V. 2011. CORN: Correlation-driven nonparametric learning approach
for portfolio selection.}

\begin{bottomstuff}
Author's addresses: B. Li, Nanyang Business School, Nanyang Technological University, Singapore;
D. Huang, Department of Mathematics, East China University of Science and Technology, and School of Computer Science, Fudan University, China;
S. Hoi, School of Computer Engineering, Nanyang Technological University, Singapore.\\
Email: \{binli, chhoi\}@ntu.edu.sg; djhuang@fudan.edu.cn.
\end{bottomstuff}

\maketitle

\section{Proof of CORN's Universal Consistency}

In this note, we give a detailed proof that the portfolio scheme {\rm CORN}~\cite{LHG11} is universal with respect to the class of all ergodic processes.
We first give a concise definition about ``universal" considered in this note.

\begin{definition}
An investment strategy ${\mathbf{B}}$ is called universal with respect a class of stationary and ergodic processes $\{\mathbf{X}_n\}_{-\infty}^{+\infty}$, if for each process in the class,
$$\mathop {\lim }\limits_{n \to \infty } \frac{1}{n}\log {S_n}(\mathbf{B}) = {W^*}\quad  almost\  surely.$$
\end{definition}

Before we give the theorem and its proof, we introduce some necessary lemmas.

\begin{lemma}\label{lemma:shift sequence}\cite{Breiman1957}
Let $Z=\{Z_i\}_{-\infty}^{\infty}$ be a stationary and ergodic process. For each positive integer i, let $T^i$ denote the operator that shifts any sequence $\{...,z_{-1},z_0,z_1,...\}$ by i digits to the left. Let $f_1, f_2,...$ be a sequence of real-valued functions such that $\lim_{n \to \infty} f_n(Z)=f(Z)$ almost surely for some function f. Assume that $\mathbb{E} \sup_n |f_n(Z)| < \infty$. Then
$$\lim \limits_{n \to \infty} \frac{1}{n} \sum \limits_{i=1}^n f_i(T^iZ)=\mathbb{E}f(Z) \quad almost\ surely.$$
\end{lemma}

\begin{lemma}\label{lemma:portfolio return}\cite{Algoet1988}
Let $\mathbf{Q}_{n \in \mathcal{N} \cup \{\infty\}}$ be a family of regular probability distributions over the set $\mathbb{R}_+^d$ of all market vectors such that $E\{|\log U_n^{(j)}|\}<\infty$ for any coordinate of a random market vector $\mathbf{U}_n=(U_n^{(1)},...,U_n^{(d)})$ distributed according to $\mathbf{Q}_n$. In addition, let $\mathbf{B}^*(\mathbf{Q}_n)$ be the set of all log-optimal portfolios with respect to $\mathbf{Q}_n$, that is, the set of all portfolios {\bf b} that attain $\max_{\mathbf{b} \in \Delta_d} E\{\log\langle \mathbf{b},\mathbf{U}_n\rangle\}$. Consider an arbitrary sequence $\mathbf{b}_n \in \mathbf{B}^*(\mathbf{Q}_n)$. If
$$\mathbf{Q}_n \to \mathbf{Q}_\infty \quad weakly\ as\ n \to \infty,$$
then, for $\mathbf{Q}_\infty$-alomst all {\bf u},
$$\lim \limits_{n \to \infty}\langle  \mathbf{b}_n,\mathbf{u}\rangle \to \langle \mathbf{b}^*,\mathbf{u}\rangle,$$
where the right-hand side is constant as ${\bf b^*}$ ranges over $B^*(\mathbf{Q}_\infty)$.
\end{lemma}

\begin{lemma}\label{lemma:expectation field}\cite{Algoet1988}
Let {\bf X } be a random market vector defined on a probability space$(\Omega,\mathcal{F},\mathbb{P})$ satisfying $E\{|\log X^{(j)}|\}<\infty$. If $\mathcal{F}_k$ is an increasing sequence of sub-$\sigma$-fields of $\mathcal{F}$ with
$$\mathcal{F}_k \nearrow \mathcal{F}_\infty \subseteq \mathcal{F},$$
then
$$\mathbb{E}\left\{\max \limits_\mathbf{b} \mathbb{E}[\log \langle \mathbf{b}, \mathbf{X}\rangle | \mathcal{F}_k]\right\} \nearrow \mathbb{E}\left\{\max \limits_\mathbf{b} \mathbb{E}[\log \langle \mathbf{b}, \mathbf{X}\rangle | \mathcal{F}_\infty]\right\},$$
as $k \to \infty$ where the maximum on the left-hand side is taken on over all $\mathcal{F}_k$-measurable functions {\bf b} and the maximum on the right-hand side is taken on over all $\mathcal{F}_\infty$-measurable functions {\bf b}.
\end{lemma}

\begin{lemma}\label{lemma:lebesgue denisty theorem}
Let $\mu$ be the Lebesgue measure on the Euclidean space $\mathbf{R}^n$ and A be a Lebesgue measurable subset of $\mathbf{R}^n$. Define the approximate density of A in a $\varepsilon$-neighborhood of a point x in $\mathbf{R}^n$ as
$$d_\varepsilon\left(x\right)=\frac{\mu\left(A\cap B_\varepsilon\left(x\right)\right)}{\mu\left(B_\varepsilon\left(x\right)\right)},$$
where $B_\varepsilon$ denotes the closed ball of radius $\varepsilon$ centered at x.
Then for almost every point x of A the density
$$d(x)=\lim \limits_{\varepsilon \to 0}d_\varepsilon(x)$$
exists and is equal to 1.
\end{lemma}

\begin{lemma}
\label{lemma:similarity}
The inequality
$$\frac{{{\mathop{\rm cov}} \left(\mathbf{X},{\mathbf{X}'}\right)}}{{\sqrt {Var\left(\mathbf{X}\right)} \sqrt {Var\left({\mathbf{X}'}\right)} }} \ge \rho,$$
which describe the similarity of $\mathbf{X}$ and $\mathbf{X}'$ in CORN strategy, is approximately equivalent to
$$2{\rm Var}\left(\mathbf{X}\right)\left(1-\rho\right) \geq \mathbb{E}\{\left(\mathbf{X}-\mathbf{X}'\right)^2\}.$$
\end{lemma}

\begin{proof}
In general, from the covariance ${\rm cov}\left(\mathbf{X},\mathbf{X}'\right)$ it is impossible to derive a topology, since ${\rm cov}\left(\mathbf{X},\mathbf{X}'\right)=1$ doesn't imply that $\mathbb{E}\left\{\left(\mathbf{X}-\mathbf{X}'\right)^2\right\}=0$.
However, because $\mathbf{X}$ and $\mathbf{X}'$ are relative prices, then we have $\mathbb{E}\left\{\left(\mathbf{X}-\mathbf{X}'\right)^2\right\}\approx 0$. For the Euclidean distance, we have that
$$\mathbb{E}\left\{\left(\mathbf{X}-\mathbf{X}'\right)^2\right\}={\rm Var}\left(\mathbf{X}-\mathbf{X}'\right)+\left(\mathbb{E}\left\{\mathbf{X}-\mathbf{X}'\right\}\right)^2 ={\rm Var}\left(\mathbf{X}\right)-2{\rm cov}\left(\mathbf{X},\mathbf{X}'\right)+{\rm Var}\left(\mathbf{X}'\right)+\left(\mathbb{E}\left\{\mathbf{X}-\mathbf{X}'\right\}\right)^2.$$
Thus, the similarity means that
$$\frac{{\rm Var}(\mathbf{X})+{\rm Var}(\mathbf{X}')+(\mathbb{E}\{\mathbf{X}-\mathbf{X}'\})^2-\mathbb{E}\{(\mathbf{X}-\mathbf{X}')^2\}}{{\sqrt {Var(\mathbf{X})} \sqrt {Var({\mathbf{X}'})} }}\geq 2\rho$$
or equivalently,
$${\rm Var}\left(\mathbf{X}\right)+{\rm Var}\left(\mathbf{X}'\right)+\left(\mathbb{E}\left\{\mathbf{X}-\mathbf{X}'\right\}\right)^2-2\rho{\sqrt {Var\left(\mathbf{X}\right)} \sqrt {Var\left({\mathbf{X}'}\right)} } \geq \mathbb{E}\left\{\left(\mathbf{X}-\mathbf{X}'\right)^2\right\}.$$
Since both ${\rm Var}\left(\mathbf{X}\right)$ and $\left|\mathbb{E}\left\{\mathbf{X}-\mathbf{X}'\right\}\right|$ have the same order of magnitude~\footnote{See Bin's thesis, Table 7.7.}, they are in the range $10^{-4}, 10^{-3}$, therefore the previous inequality approximately means that
$$2{\rm Var}\left(\mathbf{X}\right)\left(1-\rho\right) \geq \mathbb{E}\left\{\left(\mathbf{X}-\mathbf{X}'\right)^2\right\}.$$
\end{proof}

\begin{lemma}
\label{lemma:universal}
Assume that ${\bf x_1,x_2,...}$ are the realizations of the random vectors $\bf X_1,X_2,...$ drawn from the vector-valued stationary and ergodic process $\{\mathbf{X}_n\}_{-\infty}^\infty$. The fundamental limits, determined in \cite{Algoet1992,Algoet1994,Algoet1988}, reveal that the so-called {\it log-optimum portfolio} $\mathbf{B}^*=\{\mathbf{b}^*(\cdot)\}$ is the best possible choice. More precisely, in trading period {\it n} let $\mathbf{b}^*(\cdot)$ be such that
$$\mathbb{E}\{ \log \left\langle {{\mathbf{b}^*}(\mathbf{X}_1^{n - 1}),{\mathbf{X}_n}} \right\rangle |\mathbf{X}_1^{n - 1}\}  = \mathop {\max }\limits_{\mathbf{b}( \cdot)} \mathbb{E}\{ \log \left\langle {b(\mathbf{X}_1^{n - 1}),{\mathbf{X}_n}} \right\rangle |\mathbf{X}_1^{n - 1}\}.$$
If $S_n^*=S_n(\mathbf{B}^*)$ denotes the capital achieved by a log-optimum portfolio strategy $\mathbf{B}^*$, after n trading periods, then for any other investment strategy $\mathbf{B}$ with capital $S_n=S_n(\mathbf{B})$ and for any stationary and ergodic process $\{X_n\}_{-\infty}^{\infty}$,
$$\mathop {\lim \sup}\limits_{n \to \infty } \frac{1}{n}\log \frac{{{S_n}}}{{S_n^*}} \le 0 \quad almost \ surely$$
and
$$\mathop {\lim }\limits_{n \to \infty } \frac{1}{n}\log S_n^*=W^* \quad almost \ surely, $$
where
$${W^*} = \mathbb{E}\{ \mathop {\max }\limits_{\mathbf{b}( \cdot )} \mathbb{E}\{ \log \left\langle {\mathbf{b}(\mathbf{X}_{ - \infty }^{ - 1}),{\mathbf{X}_0}} \right\rangle |\mathbf{X}_{ - \infty }^{ - 1}\} \},$$
is the maximal possible growth rate of any investment strategy.
\end{lemma}

Now, we give the universal theorem and its proof.

\begin{theorem}\label{theorem:CORN Universal}
The portfolio scheme {\rm CORN} is universal with respect to the class of all ergodic processes such that $\mathbb{E}\left\{\left|\log X^{\left(j\right)}\right|\right\}<\infty$, for $j=1,2,...,d$.
\end{theorem}

\begin{proof}
To prove the strategy {\rm CORN} is universal with respect to the class of all ergodic processes, we need to prove that if for each process in the class, $$\mathop {\lim }\limits_{n \to \infty } \frac{1}{n}\log {S_n}\left(\mathbf{B}\right) = {W^*}\quad  almost\  surely,$$
where $\mathbf{B}$ denote the strategy {\rm CORN}, and
$${W^*} =\mathop{\lim}\limits_{n \to \infty} \frac{1}{n}\log S_n^* = \mathbb{E}\left\{\mathop {\max}\limits_{\mathbf{b}\left( \cdot \right)} \mathbb{E}\left\{ \log \left\langle {\mathbf{b}\left(\mathbf{X}_{ - \infty }^{ - 1}\right), {\mathbf{X}_0}} \right\rangle |\mathbf{X}_{ - \infty }^{ - 1}\right\} \right\}.$$
We divide the proof into three parts.

(i).\quad According to the Lemma~\ref{lemma:universal}, we know that $\mathop {\lim }\limits_{n \to \infty } \left(\frac{1}{n}\log {S_n} - \frac{1}{n}\log S_n^*\right) \le 0$, then $\mathop {\lim}\limits_{n \to \infty } \frac{1}{n}\log S_n \leq \mathop {\lim}\limits_{n \to \infty } \frac{1}{n}\log S_n^*=W^*$. So it suffices to prove that
$$\mathop {\lim \inf }\limits_{n \to \infty } {W_n}\left(\mathbf{B}\right) = \mathop {\lim \inf }\limits_{n \to \infty } \frac{1}{n}\log {S_n}\left(\mathbf{B}\right) \ge {W^*} \quad {\rm almost \ surely}.$$

Without loss of generality we may assume $S_0=1$, so that
\begin{eqnarray*}
  W_n(\mathbf{B}) &=& \frac{1}{n}\log S_n\left(\mathbf{B}\right) \\
                  &=& \frac{1}{n}\log \left(\mathop{\sum} \limits_{\omega,\rho} q_{\omega,\rho} S_n(\epsilon ^{(\omega,\rho)})\right)\\
             & \geq & \frac{1}{n}\log \left(\mathop{\sup} \limits_{\omega,\rho} q_{\omega,\rho} S_n(\epsilon ^{(\omega,\rho)})\right)\\
                  &=& \frac{1}{n}\mathop{\sup} \limits_{\omega,\rho}(\log q_{\omega,\rho}+\log S_n(\epsilon ^{(\omega,\rho)}))\\
                  &=& \mathop{\sup} \limits_{\omega,\rho} \left(W_n(\epsilon ^{(\omega,\rho)})+\frac{\log q_{\omega,\rho}}{n} \right).
\end{eqnarray*}
Thus
\begin{eqnarray}
  \mathop {\lim \inf} \limits_{n \to \infty}W_n(\mathbf{B}) &=&  \mathop {\lim \inf} \limits_{n \to \infty} \mathop{\sup} \limits_{\omega,\rho} \left(W_n(\epsilon ^{(\omega,\rho)})+\frac{\log q_{\omega,\rho}}{n} \right) \nonumber\\
  &\geq & \mathop{\sup} \limits_{\omega,\rho} \mathop {\lim \inf} \limits_{n \to \infty} \left(W_n(\epsilon ^{(\omega,\rho)})+\frac{\log q_{\omega,\rho}}{n} \right) \nonumber\\
  &=& \mathop{\sup} \limits_{\omega,\rho} \mathop {\lim \inf} \limits_{n \to \infty} W_n(\epsilon ^{(\omega,\rho)}).
\end{eqnarray}

The simple argument above shows that the asymptotic rate of growth of the strategy $\mathbf{B}$ is at least as large the supremum of the rates of growth of all elementary strategies $\epsilon ^{(\omega,\rho)}$. Thus, to estimate $\mathop {\lim \inf}_{n \to \infty}W_n(\mathbf{B})$, it suffices to investigate the performance of expert $\epsilon ^{(\omega,\rho)}$ on the stationary and ergodic market sequence $\mathbf{X}_0,\mathbf{X}_{-1},\mathbf{X}_{-2},...$.

(ii)\quad First, let the integers $\omega,\rho$ and the vector $\mathbf{s}=\mathbf{s}_{-\omega}^{-1} \in \mathbb{R}_{+}^{d\omega}$ be fixed. Form the Lemma~\ref{lemma:similarity}, we can get that the set $\{X_i:1-j+\omega\leq i\leq 0, \frac{{{\mathop{\rm cov}} (\mathbf{X}_{i-\omega}^{i-1},\mathbf{s})}}{{\sqrt {Var(\mathbf{X}_{i-\omega}^{i-1})} \sqrt {Var(\mathbf{s})} }} \ge \rho\}$ can be expressed as $\{X_i:1-j+\omega\leq i\leq 0, \mathbb{E}\{(\mathbf{X}_{i-\omega}^{i-1}-\mathbf{s})^2\} \leq 2{\rm Var}(\mathbf{s})(1-\rho)$.

Let $\mathbb{P}_{j,\mathbf{s}}^{(\omega,\rho)}$ denote the (random) measure concentrated on $\{X_i:1-j+\omega\leq i\leq 0, \mathbb{E}\{(\mathbf{X}_{i-\omega}^{i-1}-\mathbf{s})^2\} \leq 2{\rm Var}(\mathbf{s})(1-\rho)$ defined by
$$\mathbb{P}_{j,\mathbf{s}}^{(\omega,\rho)}(A)= \frac{\sum \limits_{i:1-j+\omega\leq i\leq 0, \mathbb{E}\{(\mathbf{X}_{i-\omega}^{i-1}-\mathbf{s})^2\} \leq 2{\rm Var}(\mathbf{s})(1-\rho)}II_A(\mathbf{X}_i)}{|\{i:1-j+\omega\leq i\leq 0, \mathbb{E}\{(\mathbf{X}_{i-\omega}^{i-1}-\mathbf{s})^2\} \leq 2{\rm Var}(\mathbf{s})(1-\rho)\}|},\quad A \subset \mathbb{R}_+^d$$
where $II_A$ denotes the indictor of function of the set $A$. If the above set of $\mathbf{X}_i's$ is empty, then let $\mathbb{P}_{j,\mathbf{s}}^{(\omega,\rho)}=\delta_{(1,...,1)}$ be the probability measure concentrated on the vector $(1,...,1)$. In other words, $\mathbb{P}_{j,\mathbf{s}}^{(\omega,\rho)}(A)$ is the relative frequency of the vectors among $\mathbf{X}_{1-j+\omega},...,\mathbf{X}_0$ that fall in the set $A$.

Observe that for all $\mathbf{s}$, without probability one,
$$\mathbb{P}_{j,s}^{(\omega,\rho )} \to \mathbb{P}_s^{*(\omega,\rho )} =
\left\{ \begin{array}{ll}
 {\mathbb{P}_{ {\mathbf{X}_0}|\mathbb{E}\{ {{(\mathbf{X}_{- \omega}^{- 1} - \mathbf{s})}^2}\}  \le 2Var(\mathbf{s})(1-\rho) }}& {\rm if}\ \mathbb{P}(\mathbb{E}\{ {(\mathbf{X}_{ - \omega}^{ - 1} - \mathbf{s})^2}\}  \le 2Var(\mathbf{s})(1 - \rho )) > 0 \\
 {\delta _{(1,...,1)}}& {\rm if}\ \mathbb{P}(\mathbb{E}\{ {(\mathbf{X}_{ - \omega}^{ - 1} - \mathbf{s})^2}\}  \le 2Var(\mathbf{s})(1 - \rho )) = 0
\end{array} \right. \eqno(2)$$
weakly as $j \to \infty$ where $\mathbb{P}_s^{*(\omega,\rho )}$ denote the limit distribution of $\mathbb{P}_{j,s}^{(\omega,\rho )}$, $\mathbb{P}_{ {\mathbf{X}_0}|\mathbb{E}\{ {{(\mathbf{X}_{- \omega}^{- 1} - \mathbf{s})}^2}\}  \le 2Var(\mathbf{s})(1-\rho) }$ denotes the distribution of the vector $\mathbf{X}_0$ conditioned on the event $\mathbb{E}\{ {{(\mathbf{X}_{- \omega}^{- 1} - \mathbf{s})}^2}\}\le 2Var(\mathbf{s})(1 - \rho ))$. To see this, let $f$ be a bounded continuous function defined on $\mathbb{R}_+^d$. Then the ergodic theorem implies that if $\mathbb{P}\left(\mathbb{E}\{ {(\mathbf{X}_{ - \omega}^{ - 1} - \mathbf{s})^2}\}  \le 2Var(\mathbf{s})(1 - \rho )\right) > 0$, then
$$\begin{array}{cl}
  \int f(\mathbf{x})\mathbb{P}_{j,s}^{*(\omega,\rho )}(d\mathbf{x})&= \frac{\frac{1}{|1-j+\omega|}\sum \limits_{i:1-j+\omega\leq i\leq 0, \mathbb{E}\{(\mathbf{X}_{i-\omega}^{i-1}-\mathbf{s})^2\} \leq 2{\rm Var}(\mathbf{s})(1-\rho)}f(\mathbf{X}_i)}{\frac{1}{|1-j+\omega|}|\{i:1-j+\omega\leq i\leq 0, \mathbb{E}\{(\mathbf{X}_{i-\omega}^{i-1}-\mathbf{s})^2\} \leq 2{\rm Var}(\mathbf{s})(1-\rho)\}|} \\
   &\to \frac{E\{f(X_0)II_{\{\mathbb{E}\{ {{(\mathbf{X}_{- \omega}^{- 1} - \mathbf{s})}^2}\}  \le 2Var(\mathbf{s})(1-\rho) \}}\}}{\mathbb{P}\{\mathbb{E}\{ {{(\mathbf{X}_{- \omega}^{- 1} - \mathbf{s})}^2}\}  \le 2Var(\mathbf{s})(1-\rho)\}} \\
   &= E\{f(X_0)|\mathbb{E}\{ {{(\mathbf{X}_{- \omega}^{- 1} - \mathbf{s})}^2}\}  \le 2Var(\mathbf{s})(1-\rho) \}\\
   &= \int{ f(\mathbf{x})\mathbb{P}_{ {\mathbf{X}_0}|\mathbb{E}\{ {{(\mathbf{X}_{- \omega}^{- 1} - \mathbf{s})}^2}\}  \le 2Var(\mathbf{s})(1-\rho) }}\quad {\rm almost\ surely,\ as}\ j \to \infty.
\end{array}$$
On the other hand, if $\mathbb{P}\left(\mathbb{E}\{ {(\mathbf{X}_{ - \omega}^{ - 1} - \mathbf{s})^2}\}  \le 2Var(\mathbf{s})(1 - \rho )\right) = 0$, then with probability one $\mathbb{P}_{j,\mathbf{s}}^{(\omega,\rho)}$ is concentrated on $(1,...,1)$ for all j, and
$$\int f(\mathbf{x})\mathbb{P}_{j,s}^{(\omega,\rho )}(d\mathbf{x})=f(1,...,1).$$

Recall that by definition, $\mathbf{b}^{(\omega,\rho)}(\mathbf{X}_{1-j}^{-1},\mathbf{s})$ is a log-optimal portfolio with respect to the probability measure $\mathbb{P}_{j,\mathbf{s}}^{(\omega,\rho)}$. Let $\mathbf{b}_{\omega,\rho}^*(\mathbf{s})$ denote a log-optimal portfolio with a respect to the limit distribution $\mathbb{P}_s^{*(\omega,\rho)}$. Then, using Lemma~\ref{lemma:portfolio return}, we infer from equation~(2) that, as $j$ tends to infinity, we have the almost sure convergence
$$\mathop {\lim }\limits_{j \to \infty } \left\langle \mathbf{b}^{(\omega,\rho)}(\mathbf{X}_{1-j}^{-1},\mathbf{s}),\mathbf{x}_0 \right\rangle=\left\langle \mathbf{b}_{\omega,\rho}^*(\mathbf{s}),\mathbf{x}_0 \right\rangle,$$
for $\mathbb{P}_s^{*(\omega,\rho)}$ -almost all $\mathbf{x}_0$ and hence for $\mathbb{P}_{\mathbf{X}_0}$-almost all $\mathbf{x}_0$. Since {\bf s} was arbitrary, we obtain
$$\mathop {\lim }\limits_{j \to \infty } \left\langle \mathbf{b}^{(\omega,\rho)}(\mathbf{X}_{1-j}^{-1},\mathbf{\mathbf{X}_{-\omega}^{-1}}),\mathbf{x}_0 \right\rangle=\left\langle \mathbf{b}_{\omega,\rho}^*(\mathbf{X}_{-\omega}^{-1}),\mathbf{x}_0 \right\rangle \quad {\rm almost}\ {\rm surely}, \eqno(3)$$

Next, we apply Lemma~\ref{lemma:shift sequence} for the function
$$f_i(\mathbf{x}_{-\infty}^{\infty})=\log\left\langle \mathbf{h}^{(\omega,\rho)}(\mathbf{x}_{1-i}^{-1}),\mathbf{x}_0\right\rangle=\log\left\langle \mathbf{b}^{(\omega,\rho)}(\mathbf{x}_{1-i}^{-1},\mathbf{\mathbf{x}_{-\omega}^{-1}}),\mathbf{x}_0 \right\rangle$$
defined on $\mathbf{x}_{-\infty}^{\infty}=(...,\mathbf{x}_{-1},\mathbf{x}_0,\mathbf{x}_{1})$. Note that
$$\left| f_i(\mathbf{X}_{-\infty}^{\infty}) \right|=\left|\log\left\langle \mathbf{h}^{(\omega,\rho)}(\mathbf{X}_{1-i}^{-1}),\mathbf{x}_0\right\rangle \right| \leq \sum \limits_{j=1}^d \left|\log X_0^{(j)}\right|,$$
which has finite expectation, and
$$f_i(\mathbf{X}_{-\infty}^{\infty}) \to \left\langle b_{\omega,\rho}^*(X_{-\omega}^{-1}),X_0 \right\rangle \quad {\rm almost \ surely \ as}\ i \to \infty$$
by equation~(3). As $n \to \infty$, Lemma~\ref{lemma:shift sequence} yields
$$\begin{array}{cl}
  W_n(\mathbf{\epsilon}^{(\omega,\rho)}) &= \frac{1}{n}\sum \limits_{i=1}^n \log \left\langle \mathbf{h}^{(\omega,\rho)}(\mathbf{X}_1^{i-1}),\mathbf{X}_i \right\rangle \\
   &= \frac{1}{n}\sum \limits_{i=1}^n f_i(T^iX_{-\infty}^{\infty}) \\
   &\to \mathbb{E}\left\{ \log \mathbf{b}_{\omega,\rho}^*(X_{-\omega}^{-1}),X_0\right\} \\
   &\mathop{=}\limits^{def} \theta_{\omega,\rho} \quad {\rm almost \ surely}.
\end{array}$$

Therefore, by equation~(1) we have
$$\mathop {\lim \inf} \limits_{n\to \infty}W_n(\mathbf{B})\geq \sup \limits_{\omega,\rho}\theta_{\omega,\rho} \geq \sup \limits_{\omega}      \mathop {\lim \inf} \limits_{\rho}\theta_{\omega,\rho} \quad {\rm almost \ surely},$$
and it suffices to show that the right-hand side is at least $W^*$.

(iii)\quad To this end, first, define, for Borel sets $A,B \subset \mathbb{R}_+^d$,
$$m_A(z)=\mathbb{P}\left\{ \mathbf{X}_0 \in A | \mathbf{X}_{-\omega}^{-1}=z\right\}$$
and
$$\mu_\omega(B)=\mathbb{P}\left\{\mathbf{X}_{-\omega}^{-1}\in B\right\}.$$

Then, for any $\mathbf{s} \in {\rm support}(\mu_{\omega})$, and for all $A$,
$$\begin{array}{cl}
    \mathbb{P}_s^{*(\omega,\rho)}(A) & = \mathbb{P}\{{ {\mathbf{X}_0} \in A|\mathbb{E}\{ {{(\mathbf{X}_{- \omega}^{- 1} - \mathbf{s})}^2}\}  \le 2Var(\mathbf{s})(1-\rho) }\} \\
    &= \frac{\mathbb{P}\{{ {\mathbf{X}_0} \in A,\mathbb{E}\{ {{(\mathbf{X}_{- \omega}^{- 1} - \mathbf{s})}^2}\}  \le 2Var(\mathbf{s})(1-\rho) }\}}{\mathbb{P}\{{ \mathbb{E}\{ {{(\mathbf{X}_{- \omega}^{- 1} - \mathbf{s})}^2}\}  \le 2Var(\mathbf{s})(1-\rho) }\}}\\
    &= \frac{1}{\mu_\omega(S_{s,2Var(\mathbf{s})(1-\rho)})} \int_{S_{s,2Var(\mathbf{s})(1-\rho)}} m_A(z)\mu_\omega(dz)\\
    &\to m_A(\mathbf{s})=\mathbb{P}\{X_0 \in A|\mathbf{X}_{-\omega}^{-1}=\mathbf{s}\}
  \end{array}
$$
as $\rho \to 1$ and for $\mu_\omega$-almost all {\bf s} by Lebesgue density theorem(see Lemma~\ref{lemma:lebesgue denisty theorem} or see \cite{Gyorfi2002} Lemma24.5), and therefore
$$\mathbb{P}_{\mathbf{X}_{-\omega}^{-1}}^{*(\omega,\rho)}(A) \to \mathbb{P}\{ \mathbf{X}_0 \in A|\mathbf{X}_{-\omega}^{-1}\}$$
as $\rho \to 1$ for all $A$. Thus, using Lemma~\ref{lemma:portfolio return} again, we have
$$\begin{array}{ccl}
    \mathop{\lim \inf} \limits_{\rho} \theta_{\omega,\rho} & = & \lim \limits_{\rho} \theta_{\omega,\rho} \\
     & = & \lim \limits_{\rho}\mathbb{E}\left\{ \log \mathbf{b}_{\omega,\rho}^*(X_{-\omega}^{-1}),X_0\right\}\\
     & = & \mathbb{E}\{\log \left\langle \mathbf{b}_\omega^*(X_{-\omega}^{-1}),\mathbf{X}_0\right\rangle \} \\
     &   & {(\rm where}\ \mathbf{b}_\omega^*(\cdot)\ {\rm is\ the\ log}-{\rm optimum\ portfolio\ with\ respect} \\
     &   &{\rm to\ the\ conditional\ probability\ } {\mathbb{P}\{\mathbf{X}_0 \in A|\mathbf{X}_{-\omega}^{-1}\})} \\
     & = & \mathbb{E}\left\{\max \limits_{\mathbf{b}(\cdot)}\mathbb{E}\{\log \left\langle \mathbf{b}(\mathbf{X}_{-\omega}^{-1}),\mathbf{X}_0\right\rangle|\mathbf{X}_{-\omega}^{-1}\}\right\} \\
     & = & \mathbb{E}\left\{\mathbb{E}\{\log \left\langle \mathbf{b}_{\omega}^*(\mathbf{X}_{-\omega}^{-1}),\mathbf{X}_0\right\rangle|\mathbf{X}_{-\omega}^{-1}\}\right\} \\
     & \mathop {=} \limits^{def} & \theta_\omega^*.
  \end{array}
$$

Next, to finish the proof, we appeal to the submartingale convergence theorem. First note the sequence
$$Y_\omega \mathop{=} \limits^{def}\mathbb{E}\{\log \left\langle \mathbf{b}_{\omega}^*(\mathbf{X}_{-\omega}^{-1}),\mathbf{X}_0\right\rangle|\mathbf{X}_{-\omega}^{-1}\}=\max \limits_{\mathbf{b}(\cdot)}\mathbb{E}\{\log \left\langle \mathbf{b}(\mathbf{X}_{-\omega}^{-1}),\mathbf{X}_0\right\rangle|\mathbf{X}_{-\omega}^{-1}\}$$
of random variables forms a submartingale, that is, $\mathbb{E}\{Y_{\omega+1}|Y_{-\omega}^{-1} \geq Y_\omega\}$. To see this, note that
$$
\begin{array}{cl}
  \mathbb{E}\{Y_{\omega+1}|\mathbf{X}_{-\omega}^{-1}\} & =\mathbb{E}\{\mathbb{E}\{\log \left\langle \mathbf{b}_{\omega+1}^*(\mathbf{X}_{-\omega-1}^{-1}),\mathbf{X}_0\right\rangle|\mathbf{X}_{-\omega-1}^{-1}\}|\mathbf{X}_{-\omega}^{-1}\} \\
  & \geq \mathbb{E}\{\mathbb{E}\{\log \left\langle \mathbf{b}_{\omega}^*(\mathbf{X}_{-\omega}^{-1}),\mathbf{X}_0\right\rangle|\mathbf{X}_{-\omega-1}^{-1}\}|\mathbf{X}_{-\omega}^{-1}\}\\
  &=\mathbb{E}\{\log \left\langle \mathbf{b}_{\omega}^*(\mathbf{X}_{-\omega}^{-1}),\mathbf{X}_0\right\rangle|\mathbf{X}_{-\omega-1}^{-1}\}\\
  &=Y_\omega.
\end{array}
$$
This sequence is bounded by
$$\max \limits_{\mathbf{b}(\cdot)}\mathbb{E}\{\log \left\langle \mathbf{b}(\mathbf{X}_{-\infty}^{-1}),\mathbf{X}_0\right\rangle|\mathbf{X}_{-\infty}^{-1}\},$$
which has a finite expectation. The submartingale convergence theorem(see Stout1974) implies that submartingale is convergence almost surely, and $\sup_\omega \theta_\omega^*$ is finite. In particular, by the submartingale property, $\theta_\omega^*$ is a bounded increasing sequence, so that
$$\sup \limits_\omega \theta_\omega^*=\lim \limits_{\omega \to \infty} \theta_\omega^*.$$

Applying Lemma~\ref{lemma:expectation field} with the $\sigma-$algebras
$$\sigma(\mathbf{X}_{-\omega}^{-1}) \nearrow \sigma(\mathbf{X}_{-\infty}^{-1})$$
yields
$$
\begin{array}{cl}
  \sup \limits_\omega \theta_\omega^* & =\lim \limits_{\omega \to \infty} \mathbb{E}\left\{\max \limits_{\mathbf{b}(\cdot)}\mathbb{E}\{\log \left\langle \mathbf{b}(\mathbf{X}_{-\omega}^{-1}),\mathbf{X}_0\right\rangle|\mathbf{X}_{-\omega}^{-1}\}\right\} \\
   & = \mathbb{E}\left\{\max \limits_{\mathbf{b}(\cdot)}\mathbb{E}\{\log \left\langle \mathbf{b}(\mathbf{X}_{-\infty}^{-1}),\mathbf{X}_0\right\rangle|\mathbf{X}_{-\infty}^{-1}\}\right\}\\
   & =W^*.
\end{array}
$$
Then
$$\mathop {\lim \inf} \limits_{n\to \infty}W_n(\mathbf{B})\geq \sup \limits_{\omega,\rho}\theta_{\omega,\rho} \geq \sup \limits_{\omega}      \mathop {\lim \inf} \limits_{\rho}\theta_{\omega,\rho}= \sup \limits_\omega \theta_\omega^*=W^* \quad {\rm almost \ surely},$$
and from the above three parts of proof, we can get that
$$\mathop {\lim }\limits_{n \to \infty } \frac{1}{n}\log {S_n}\left(\mathbf{B}\right) = {W^*}\quad  {\rm almost\  surely}$$
and the proof of Theorem~\ref{theorem:CORN Universal} is finished.
\end{proof}

\bibliographystyle{ACM-Reference-Format-Journals}
\bibliography{CORN_Proof}

\end{document}